\documentclass{article}
\usepackage{colt09e}
\usepackage{times}
\usepackage{epsfig}%,amsmath,amsthm}%,amssymb} latexsym,amsfonts,mathrsfs}%, 
\usepackage{amsmath}
\usepackage[noeepic]{qtree}

\newcommand{\abs}[1]{\lvert #1 \rvert}
\newcommand{\rank}{\operatorname{rank}}
\newcommand{\union}{\cup}    
\newcommand{\Union}{\bigcup} 
  
\newcommand{\Intersect}{\bigcap}
\newcommand{\card}[1]{\abs{#1}}
\newcommand{\cC}{\mathcal{C}}
\newcommand{\cI}{\mathcal{I}}
\newcommand{\cO}{\mathcal{O}}
\newcommand{\defeq}{\,:=\,}    
\renewcommand{\th}{\ifmmode{^{\textrm{th}}}\else{\textsuperscript{th}\ }\fi}
\newcommand{\comment}[1]{}

\begin{document}
\title{Matroids Hitting Sets and Unsupervised Dependency Grammar Induction}

%\author{Newton}
%\address{University 1}
%\email{email 1}

\author{
\large Nicholas Harvey$^1$ %\\ Mathematics\\University of Waterloo 
\And \large  David Karger$^3$ %\\  Computer Science\\MIT
 \And \large Vahab Mirrokni$^4$ %\\  Google 
 \And \large  Virginia Savova$^2$ %\\ Systems Biology\\Harvard Medical School  
 \And \large Leonid Peshkin$^{2,*}$ %\\Systems Biology\\Harvard Medical School
}

\maketitle

%{\bf Affiliations:} \\
{\tiny $^1$ - Dept of Computer Science, 
Univ of British Columbia, Vancouver, Canada \\

$^2$ - Systems Biology Dept, Harvard Medical School, Boston, USA \\

$^3$ - Dept of Computer Science, MIT, Boston, USA

$^4$ - Google

$^*$ - corresponding author peshkin@gmail.com }
\\

\begin{abstract}
This paper formulates a novel problem on graphs: find the minimal subset of edges in a fully connected graph, such that the resulting graph contains all spanning trees for a set of specified subgraphs. This formulation is motivated by an unsupervised grammar induction problem from computational linguistics. We present a reduction to some known problems and algorithms from graph theory, provide computational complexity results, and describe an approximation algorithm. 

\end{abstract}

%%%%%%%%%%%%%%%%%%%%%%%%%%%%%%%%%%%%%%%%%%%%%%%
\section{Introduction}

%We consider an interesting optimization problem on graphs which involves
%finding a subset of the edges that contains a spanning tree for certain induced subgraphs.
%This problem is motivated by a practical application from  the field of computational linguistics.

%%%%%%%%%%%%%%%%%%%%%%%%%%%%%%%%%%%%%%%%%%%%%%%%%%%%%%%
%\section{Dependency grammar} 

Linguistic representations of natural language syntax arrange syntactic dependencies among the words in a sentence into a tree structure, of which the string is a one dimensional projection. We are concerned with the task of analyzing a set of several sentences, looking for the most parsimonious set of corresponding syntactic structures, solely on the basis of co-occurrence of words in sentences. 
We proceed by first presenting an example, then providing a general formulation of dependency structure and grammar induction. 

\begin{figure}[b]
%\begin{flushleft}
\centerline{
\vspace{0.1cm}
\includegraphics[width=9cm]{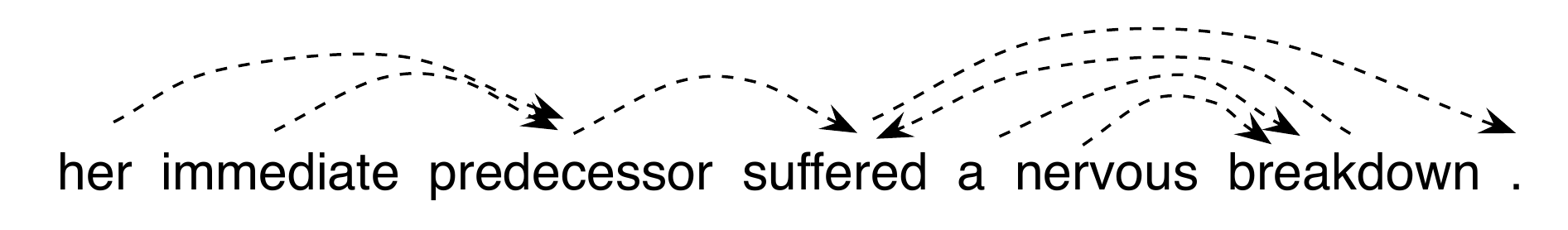}
}
%\end{flushleft}
\caption{An illstration of a dependency structure.} 
\label{fig:enc}
\end{figure} 

Consider a sentence {\em "Her immediate predecessor suffered a nervous breakdown.}"
A dependency grammar representation of this sentence shown in Figure~\ref{fig:enc} captures dependency between the subject, the object and the verb, as well as dependency between the determiner and the adjectives and their respective nouns.  
In this sentence, the subject \emph{predecessor} and the object \emph{breakdown} are related to the verb \emph{suffered}. The verb \emph{suffered} is the root of the dependency structure, that is illustrated in the diagram by a link to the period.  
Figure \ref{fig:tree} left represents the same dependency structure in a different way by ignoring the direction. Instead the dependence is related to the relative depth in the tree. 

In a dependency tree, each word is the mother of its dependents, otherwise known as their {\sc  head}. To linearize the dependency tree in Figure \ref{fig:tree}.left into a string, we introduce the dependents recursively next to their heads: \\
iteration 1: suffered \\  
iteration 2: predecessor suffered breakdown \\ 
iteration 3: her predecessor suffered a breakdown. 

Dependency and the related link grammars have received a lot of attention in the field of computational linguistics in recent years, since
these grammars enable much easier parsing than alternatives that are more complex lexicalized parse structures.
There are applications to such popular tasks as machine translation and information retrieval. 
However, all of the work is concerned with parsing, i.e. inducing a parse structure given a corpus and a grammar, rather than with grammar induction. Some work is concerned with inducing parameters of the grammar from annotated corpora, for example see work by Eisner
on dependency parsing~\cite{eisner96} or more recent work by McDonald et al.~\cite{Pereira04} and
references therein.  It has been pointed out~\cite{Pereira04} that parsing with dependency grammars
is related to Minimal Spanning Tree algorithms in general and in particular Chu-Liu-Edmonds MST algorithm was applied to dependency parsing. 

An established computational linguistics textbook has the following to say on the subject~\cite{manning99foundations}: 
{\it "... doing grammar induction from scratch is still a difficult, largely unsolved problem, and hence much emphasis has been placed on learning from bracketed corpora."} 
If grammar is not provided to begin with, parsing has to be done concurrently with learning the grammar. In the presence of grammar, among all the possibilities one needs to pick a syntactic structure consistent with the grammar. In the absence of grammar, it makes sense to appeal  to Occam's razor principle and look for the minimal set of dependencies which are consistent among themselves. 

More formally, a dependency grammar consists of a lexicon of terminal symbols (words), and an inventory of dependency relations specifying inter-lexical requirements. A string is generated by a dependency grammar if and only if:  
\begin{itemize}
	\item \vspace{-.2cm} Every word but one (ROOT) is dependent on another word. 
	\item \vspace{-.2cm} No word is dependent on itself either directly or indirectly. 
		\item \vspace{-.2cm} No word is directly dependent on more than one word. 
		\item \vspace{-.2cm} Dependencies do not cross.
\end{itemize}

\begin{figure}[bt]
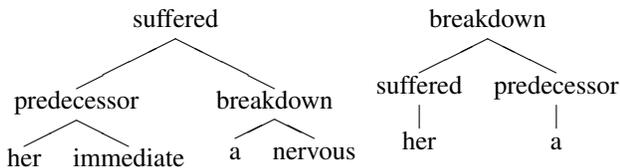

\centerline{
\qtreecenterfalse
\Tree [.{suffered} [.{predecessor} her immediate ] [.{breakdown} a nervous ] ]
\hspace{-.502cm}
\Tree [.{breakdown} [.{suffered} her ] [.{predecessor} a ] ]
}
\par
\caption{{\bf left:} A projective dependency structure for a sample sentence; {\bf right:} an example of incorrect structure of a sample sentence, also illustrating non-projective structure. }
\label{fig:tree}
\end{figure}

Unlike the first three constraints, the last constraint is a linearization constraint, usually introduced to simplify the structure and is empirically problematic. The structure in figure~\ref{fig:tree}.left is an example of so-called projective parse, in which dependency links mapped onto the sentences word sequence do not cross. Figure~\ref{fig:tree}.right illustrates an incorrect parse of the sentence with non-projective dependancies: "her"$\rightarrow$"suffered" is crossing "a"$\rightarrow$"predecessor").  While the vast majority of English sentences observe the projectivity constraint, other languages allow much more flexibility in word order. Non-projective structures include wh-relative clauses \cite{Pike43}, parentheticals \cite{McCawley82}, cross-serial constructions of the type found in Dutch and Swiss-German \cite{Ojeda88}, as well as free or relaxed word order languages \cite{Pullum82}. Therefore, it is interesting whether grammar induction can be performed without regard to word order.
 
A truly cross-linguistic formulation of dependency parsing corresponds to finding a spanning tree (parse) in a completely connected subgraph of word nodes and dependency edges.  The grammar induction problem in the same setting corresponds to inducing the minimal fully-connected subgraph which contains spanning trees for all sentences in a given corpus.
Consider three sentences: "Her immediate predecessor suffered a nervous breakdown.",  "Her predecessor suffered a stroke.", "It is a nervous breakdown."
Intuitively, the repetition of a word cooccurrence informs us about grammatical co-dependence.

Here is a formulation of the grammar induction problem as an optimization problem: 
Given a lexicon $V$ and a set of $k$ sentences $S_1,\ldots ,S_k$ s.t. $S_i \subset V$ (a.k.a. {\em corpus}) 
the objective is to find the most parsimonious combination of dependency structures. i.e. such set of spanning trees for all $S_i$  that has the minimal cardinality of a joint set of edges. 

In section~\ref{sptree} of this paper, we formally introduce the related graph-theoretic problem. In section~\ref{lognhard} we show that the problem is hard to approximate within a factor of $c \log n$ for weighted instances, and hard to approximate within some constant factor (APX-hard) for unweighed instances.
In section~\ref{matroids}, we generalize the problem to matroids.
Here we prove that the problem is hard to approximate within a factor
of $c \log n$, even for unweighed instances.
We conclude with a positive result -- an algorithm for the matroid problem which constructs a solution whose cardinality is within $O(\log n)$ of optimal.

%%%%%%%%%%%%%%%%%%%%%%%%%%%%%%%%%%%%%%%%%%%%%%%%%%%%%%%
%\subsection{Recovering links from ROOT}

%A dependency grammar G is a tuple \{V, ROOT,L\}, where V is a finite set of symbols, ROOT is a special symbol, and L is a set of ordered pairs, \emph{links} - $<w_i, w_j>, w_i \in V, w_j \in V \cup ROOT$.

%Let S be a set of ordered pairs $<x, y>_i$, where $1\leq i \leq number of pairs$. \\
%$Left(S)={x_i}$ where $1\leq i \leq number of pairs$ and $x_i$ is the first member of pair $i$. \\ 
 %$Right(S)={y_i}$, where $1\leq i \leq number of pairs$ and $y_i$ is the second member of pair $i$.\\
%$MAP(x) = y iff <x, y> \in S$

%A dependency parse P of a sentence S, given a dependency grammar G, is a set of ordered pairs $<w_i, w_j>$ such that:\\
%$<w_i, w_j> \in L_G$\\
%Each $w \in S$ is in $Left(P)$ exactly once\\
%$ROOT \in Right(P)$ exactly once\\
%If $w \in Right(P)$ then $w \in S or w = ROOT$\\
%$MAP*(w_j)\neq w_i$ \\

%Let P be a dependency parse and P' be a set of unordered pairs such that for every pair $p' \in P'$ there is an ordered pair $p in P$ and $p'$ and $p$ contain the same elements. Then $P$ can be recovered from $P'$. \\
%Proof:\\
%Define ``neighbor of'': a is a neighbor of b iff $\{a,b\} \in P'$ \\
%Define ``frontier'' $F= \{ROOT\}$\\
%Iterate $F= F \cup {w_i} w_i$ a neighbor of some $w_j \in F$.\\
%Define the frontier as $F \subset S$

%%%%%%%%%%%%%%%%%%%%%%%%%%%%%%%%%%%%%%%%%%%%%%%%%%%

\begin{figure}[tb]
\centerline{\includegraphics[width=8.5cm,angle=0]{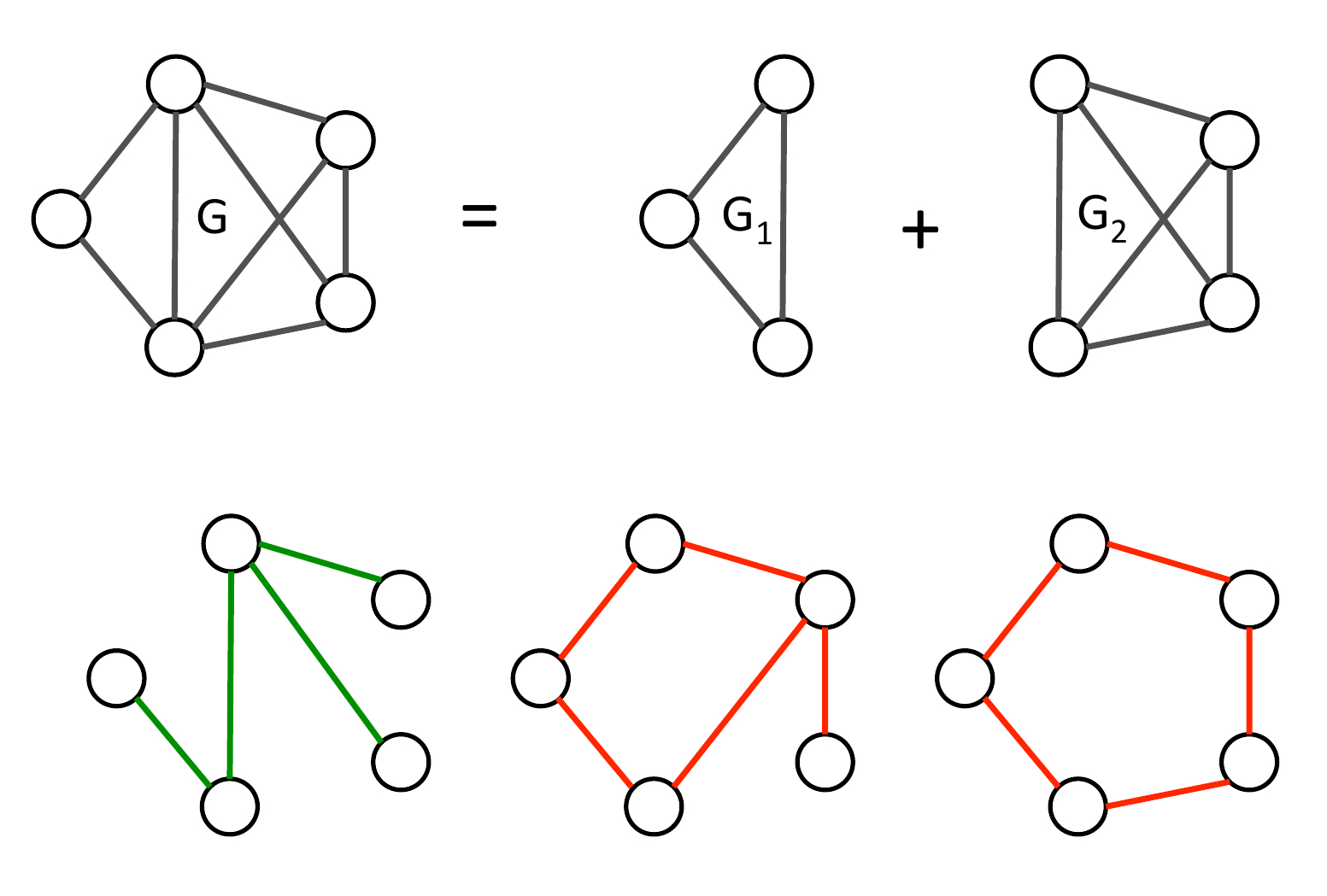}}
\caption{{\bf top:}
An instance of a problem for a graph $G$ consisting of two sub-graphs $G_1$ and $G_2$; {\bf bottom:} examples of a correct solution on the left (green) and two incorrect solutions on the right (red).  }
    \label{MSFexample}
\end{figure}

\section{The Problem for Spanning-Trees}
\label{sptree}

Let $G=(V,E)$ be a graph and let $S_1, \ldots, S_k$ be arbitrary subsets of $V$.
Our objective is to find a set of edges $F \subseteq E$ such that 
\begin{itemize}
\item $F$ contains a spanning tree for each induced subgraph $G[S_i]$, and
\item $|F|$ is minimized.
\end{itemize}
We call this the Min Spanning-Tree Hitting Set problem. Figure~\ref{MSFexample} illustrates one instance of this problem.  A graph $G$ consist of two sub-graphs $G_1$ and $G_2$. We present one possible correct solution on the left ($|F|$ = 4) and two sample incorrect solutions ($|F|$ = 5) on the right.  
The Min Spanning-Tree Hitting Set problem may be generalized to include a weight function $w$ on the edges of $G$.
The objective for the weighted problem is the same as before, except that we seek to minimize $w(F)$.
Notice that the problem initially appears similar to the group Steiner problem~\cite{Hauptmann2013ACO}, since the objective is to connect certain subsets of the nodes.
However, our condition on the subgraph is slightly different:
we require that the given subsets of nodes are \emph{internally} connected.

To develop some intuition for this problem, let's analyze a simple greedy ad-hoc solution: first, assign all the edges weight equivalent to the number of sub-graphs it is included into, i.e. count the frequency of node pairs in the input set; then fragment the graph into subgraphs, keeping the weights and run the standard MST algorithm, to find a spanning tree for each subgraph. 
Figure~\ref{txt4dcmp} presents a counterexample to simple heuristics approaches. The following sub-sets make up the input as indicated via edges of distinct color and pattern in the figure: $\{ 1,4,5 \}, \{ 2,4,5 \},$ $\{ 3,4,5 \}, \{ 1,4 \},$ $\{ 1,5 \}$, $\{ 2,4 \}, \{ 2,5 \},$ $\{ 3,4\}, \{ 3,5 \}$.  The optimal solution to this instance does not contain the edge $\{ 4,5 \}$, yet this edge is a member of the most (namely three) sub-sets.

\begin{figure}[bht]
\centerline{\includegraphics[width=6cm,angle=0]{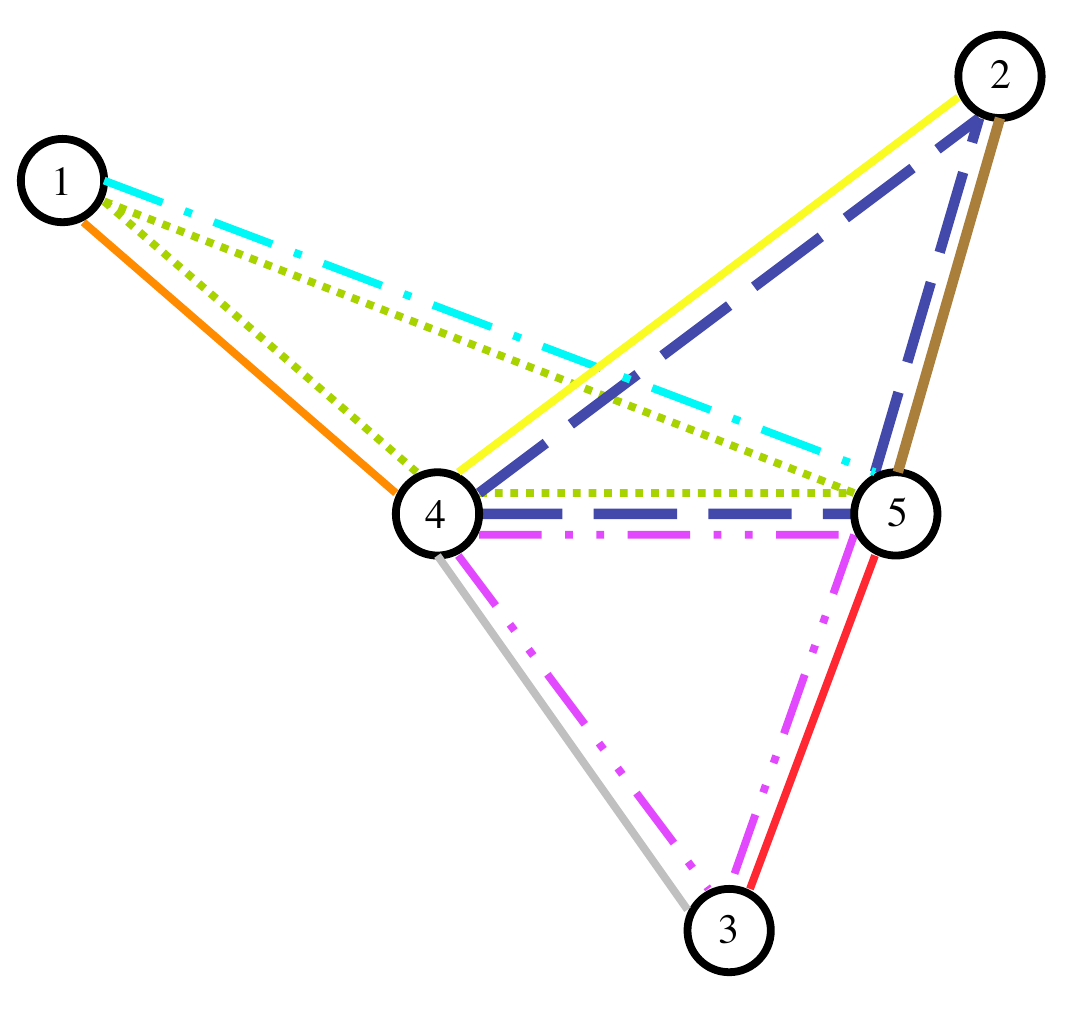}}
\caption{A counterexample to  simple heuristics approaches.}\label{txt4dcmp}
\end{figure}

\subsection{Hardness for Weighted Instances}
\label{lognhard}

We now show that the weighted problem is NP-hard to approximate within a factor of $\log n$.
To do so, we exhibit a reduction from Min Hitting Set, which is known to be hard to approximate
within $\log n$.

An instance of Min Hitting Set consists of a universe $U = \{u_1, \ldots, u_n\}$ and a collection
of sets $\mathcal{T} = \{T_1, \ldots, T_m\}$, each of which is a subset of $U$.
We construct a weighted instance of Min Spanning-Tree Hitting Set as follows.
Let $r \not \in V$ be a new vertex. We set
\begin{eqnarray*}
V =& U + r \\
E =& K_U \cup \left\{ \; \{r,u_i\} \, : \,   \mbox{for all $u_i \in U$}  \; \right\}     \\
S_{\{i,j\}} =& \{u_i,u_j\} \, \,  \,  \,  \mbox{for $1 \leq i < j \leq n$} \\
S'_i =& T_i + r \, \,  \,  \,  \mbox{for $1 \leq i \leq m$},
\end{eqnarray*}
where $K_U$ denotes (the edges of) the complete graph on vertex set $U$.
The edges belonging to $K_U$ have weight $1$ and the edges incident with $r$ have weight $n^3$.
Let $h$ denote the minimum weight of a Spanning-Tree Hitting Set in $G$.
Let $h'$ denote the minimum cardinality of a Hitting Set for $\mathcal{T}$.

\begin{claim}
\label{reduction}
$h = h' \cdot n^3 + \binom{n}{2}$.
\end{claim}
\begin{proof}
First we show that $h' \leq \big(h - \binom{n}{2} \big)/n^3$
Let $F$ be a spanning-tree hitting set.
Clearly $K_U \subseteq F$, because of the sets $S_{\{i,j\}}$.
So all edges in $F \setminus K_U$ are of the form $\{r,u_i\}$.
Now define $C = \left\{\; u_i \,:\, \{r,u_i\} \in F \;\right\}   $.
We now show that $C$ is a hitting set. Consider a set $T_i$.
Since $F$ contains a spanning tree for $S'_i$, it must contain some edge $\{r,u_i\}$.
This shows that $C$ hits the set $T_i$.

Now we show that $h \leq h' \cdot n^3 + \binom{n}{2}$.
Let $C \subseteq U$ be a hitting set for $\mathcal{T}$.
Let $F = K_U \cup \left\{\; \{r,u_i\} \,:\, \mbox{for all $u_i \in C$} \;\right\}$.
We now show that $F$ is a spanning-tree hitting set.
Each set $S_{\{i,j\}}$ is clearly hit by the set $K_U$.
So consider a set $S'_i = T_i + r$.
All edges $\{u_a,u_b\}$ with $a, b \in T_i$ are contained in $K_U$.
Furthermore, since $C$ is a hitting set, there exists an element $u_a \in T_i \cap C$.
This implies that $\{r,u_a\} \in F$,
and hence $F$ contains a spanning tree for $G[S'_i]$.
\end{proof}

Given an instance $\mathcal{T}$ of Hitting Set,
it is NP-hard to decide whether $OPT(\mathcal{T}) \leq f(n)$
or $OPT(\mathcal{T}) > \alpha \log n \cdot f(n)$ for some constant $\alpha > 0$ and some function $f$.
To prove $\log n$-hardness of Min Spanning-Tree Hitting Set,
we must similarly show that for any instance $y$,
there exists a constant $\beta>1$ and a function $g$ such that
it is NP-hard to decide whether $OPT(G) \leq g(y)$
or $OPT(G) > \beta \log n \cdot g(y)$.

From our reduction, we know that it is NP-hard to distinguish between
$$OPT(G) \leq f(n) \cdot n^3 + {\textstyle \binom{n}{2}}$$
$\text{\quad or \quad}$
$$OPT(G) > \alpha \log n \cdot f(n) \cdot n^3 + {\textstyle \binom{n}{2}}. $$
Now note that
\begin{align*}
\frac{\alpha \log n \cdot f(n) \cdot n^3 + {\textstyle \binom{n}{2}}}
{f(n) \cdot n^3 + {\textstyle \binom{n}{2}}}
&= \\
\frac{\alpha \log n \cdot \big(f(n) \cdot n^3 + {\textstyle \binom{n}{2}}/(\alpha \log n) \big)}
{f(n) \cdot n^3 + {\textstyle \binom{n}{2}}} 
&= \\
\alpha \log n \cdot \Bigg( 1 - \frac{{\textstyle \binom{n}{2}}\cdot \big(1- 1/ \alpha \log n) \big)}
{f(n) \cdot n^3 + {\textstyle \binom{n}{2}}} \Bigg)
& \geq 
\beta \log n
\end{align*}
for some constant $\beta > 0$.
Letting
$g(y) = f(n) \cdot n^3 + {\textstyle \binom{n}{2}}$,
it follows that Min Spanning-Tree Hitting Set is NP-hard to approximate within $\log n$.

\subsection{Hardness for Unweighted Instances}
\label{apxhard}

We show APX-hardness for the unweighted problem via a reduction from Vertex Cover.
The approach is similar to the construction in Section~\ref{lognhard}.
Suppose we have an instance $G'=(V',E')$ of the Vertex Cover problem.
We use the fact that Vertex Cover is equivalent to Min Hitting Set where $U = E'$ and $\mathcal{T} = E'$.
The construction differs only in that $E'$ is used in place of the edge set $K_U$;
the sets $S_{\{i,j\}}$ are adjusted accordingly.
Let $h$ denote the minimum cardinality of a Spanning-Tree Hitting Set in $G$.
Let $c$ denote the minimum cardinality of a Vertex Cover in $G'$.
A claim identical to Claim~\ref{reduction} shows that $h = c + \lvert E' \rvert$.

Recall that Vertex Cover is APX-hard even for constant-degree instances;
see, e.g., Vazirani~\cite[\S 29]{Vazirani}.
So we may assume that $ \lvert E' \rvert \leq \frac{d}{2} \lvert V' \rvert$.
Given an instance $G'=(V',E')$ of Vertex Cover with degree at most some constant $d$,
it is NP-hard to decide whether $OPT(G') \leq \alpha' \card{V'}$
or $OPT(G') > \beta' \lvert V' \rvert$ for some constant $\alpha' < \beta'$.
To prove APX-hardness of Min Spanning-Tree Hitting Set,
we must similarly show that for any instance $G$,
there exists a constant $\gamma>1$ such that
it is NP-hard to decide whether $OPT(G) \leq f(G)$
or $OPT(G) > \gamma f(G)$.
From our reduction, we know that it is NP-hard to distinguish between
$$OPT(G) \leq \alpha' (\lvert V \rvert -1) + (\lvert E \rvert - \lvert V \rvert + 1) $$
$\text{\quad or \quad}$ 
$$OPT(G) > \beta' (\lvert V \rvert -1) + (\lvert E \rvert - \lvert V \rvert + 1). $$
Now note that
\begin{align*}
\frac{\beta' (\card{V}-1) + (\lvert E \rvert - \lvert V \rvert + 1)}
{\alpha' (\lvert V \rvert -1) + (\lvert E \rvert - \lvert V \rvert + 1)}
&= \\
1 + \frac{(\beta' - \alpha') (\lvert V \rvert -1)}
{\alpha' (\lvert V \rvert -1) + (\lvert E \rvert - \card{V} + 1)} 
&\geq  \\
1 +  \frac{(\beta' - \alpha') (\card{V} - 1)}
{(d/2-1 + \alpha') \card{V} + 1 - \alpha'} 
&= \\
1 +  \frac{\beta' - \alpha'}{d/2 + \alpha'} \cdot
\frac{\card{V} - 1}{\card{V} - 1},
\end{align*}
which is a constant greater than $1$.
Letting $\gamma$ be this constant, and letting
$f(y) = \alpha' (\card{V}-1) + (\card{E} - \card{V} + 1)$,
it follows that Min Spanning-Tree Hitting Set is APX-hard.

%%%%%%%%%%%%%%%%%%%%%%%%%%%%%%%%%%%%%%%%%%%%%%%%
\section{The Problem for Matroids}
\label{matroids}

The Min Spanning-Tree Hitting Set can be rephrased as a question about matroids.
Let $E$ be a ground set.
Let $M_i = (E,\cI_i)$ be a matroid for $1 \leq i \leq k$.
Our objective is to find $F \subseteq E$ such that 
\begin{itemize}
\item $F$ contains a basis for each $M_i$, and
\item $\card{F}$ is minimized.
\end{itemize}
We call this the Minimum Basis Hitting Set problem.

\subsection{Connection to Matroid Intersection}

Suppose we switch to the dual matroids.
Note that $F$ contains a basis for $M_i$ if and only $E \setminus F \in \cI_i^*$.
Then our objective to find $F' \subseteq E$ such that
\begin{itemize}
\item $F' \in \cI_i^*$ for each $i$, and
\item $\card{F'}$ is maximized.
\end{itemize}
Suppose that such a set is found, and let $F \defeq E  \setminus F'$. %\setminus
The first property implies that $F$ contains a basis for each $M_i$.
The second property implies that $\card{F}$ is minimized.
Stated this way, it is precisely the Matroid $k$-Intersection problem.
So, from the point of view of exact algorithms, Min Basis Hitting Set
and Matroid k-Intersection problems are equivalent.
However, this reduction is not approximation-preserving,
and implies nothing about approximation algorithms.

\subsection{Hardness}

\begin{theorem}
Min Basis-Hitting Set is NP-hard.
\end{theorem}
\begin{proof}
\label{basishittinghard}
We do a reduction from the well-known problem Minimum Hitting Set.
An instance of this problem consists of a family of sets $\cC = \{ C_1, \ldots, C_k \}$.
The objective is to find a set $F \subseteq E$ such that $F \cap C_i \neq \emptyset$
for each $i$.
This problem is NP-complete.

Now we reduce it to Minimum Basis Hitting Set.
For each set $C_i$, set $M_i=(E,\cI_i)$ be the matroid where
$ \cI_i = \left\{\; \{c\} \,:\, c \in C_i \;\right\} \cup \{ \emptyset \} $.
That is, $M_i$ is the rank-1 uniform matroid on $C_i$.
So a basis hitting set for these matroids corresponds precisely
to a hitting set for the the sets $\cC$.
\end{proof}

\begin{corollary}
Min Basis Hitting Set is NP-hard to approximate with $c \log n$ for some positive constant $c$.
\end{corollary}
\begin{proof}
It is well-known that Min Hitting Set is equivalent to Set Cover, and is therefore
NP-hard to approximate within $c \log n$ for some positive constant $c$.
Since reduction given in Theorem~\ref{basishittinghard} is approximation preserving,
the same hardness applies to Min Basis Hitting Set.
\end{proof}

\subsection{An Approximation Algorithm}

We consider the greedy algorithm for the Min Basis Hitting Set problem.
Let $\cO \subseteq E$ denote an optimum solution.
Let $\rank_j$ denote the rank function for matroid $M_j$
and let $r_j$ be the rank of $M_j$, i.e., $r_j = \rank_j(E)$.
Let $F_i$ denote the set that has been chosen after the $i\th$ step of the algorithm.
Initially, we have $F_0 = \emptyset$.
For $S \subseteq E$, let $P(S,e) = \sum_{j=1}^k \big(\rank_j(S+e) - \rank_j(S)\big)$;
intuitively, this is the total ``profit'' obtained, or rank that is hit, by adding $e$ to $S$.
Let $R_i$ denote $\sum_{j=1}^k \big(r_j - \rank_j(F_i)\big)$;
intuitively, if the algorithm has chosen a set $F_i$, then $R_i$ is the total amount of
``residual rank'' that remains to be hit.

Consider the $i\th$ step of the algorithm.
Let's denote the profit obtained by choosing $e_i$ by $p_i = \max_{e \not \in F_{i-1}} P(F_{i-1},e)$.
The greedy algorithm chooses an element $e_i \not \in F_{i-1}$ achieving the maximum profit.
We now analyze the efficiency of this algorithm.
Let $\cO_i$ be a minimum-cardinality set that contains $F_i$ and is a basis hitting set.

For any set $S \supseteq F_i$ and any $e \not \in S$, we have (by submodularity):
\begin{eqnarray*}
\rank_j(S + e) + \rank_j(F_i) & \leq \rank_j(F_i + e) + \rank_j(S) \\
\rank_j(S + e) - \rank_j(S) &\leq \rank_j(F_i + e) - \rank_j(F_i) \\
P(S,e) &\leq P(F_i,e) \leq p_i
\end{eqnarray*}
This implies that each edge in $\cO_i \setminus F_i$ has profit at most $p_i$.
Since $\cO_i$ must ultimately hit all of the residual rank,
but each element hits at most $p_i$, we have $R_{i-1} \leq p_i \cdot \card{\cO_i \setminus F_i}$.

Now, note that $\card{\cO_i \setminus F_i} \leq \card{\cO}$.
This is is because of the non-decreasing property of $\rank_j$:
if $\cO$ is a basis hitting set then so is $\cO \union F_i$.
This observation yields the inequality $1 \leq p_i \cdot \card{\cO} / R_{i-1}$. 
Suppose that the greedy algorithm halts with a solution of cardinality $s$.
Then we have
\begin{align*}
s &\leq \sum_{i=1}^s \frac{ p_i \cdot \card{\cO} }{ R_{i-1} }
\leq \card{\cO} \cdot \sum_{i=1}^s \frac{ p_i }{ R_{i-1} } \\
&\leq \card{\cO} \cdot \sum_{i=1}^s \: \sum_{0 \leq j < p_i} \frac{ 1 }{R_{i-1} - j}
\leq \card{\cO} \cdot \log R_0.
\end{align*}
Here, the last inequality follows from the fact that $R_i = R_{i-1} - p_i$ for $1 \leq i \leq s$.
Note that $R_0 = \sum_{j=1}^k r_j$ is the total rank of the given matroids.

\begin{table}[thb]
{\footnotesize \caption{Greedy algorithm} \rule{1.0 \linewidth}{1pt}
\begin{tabbing}
111\=111\=1111\=11\=111\= \kill
  {\bf Input:} lexicon $V$, set of $k$ sentences $S_1,\ldots ,S_k$ s.t. $S_i \subset V$ ; \\
  {\bf Initialize:} Assign each pair of words a count of sentences it appears in; \\
  \> Sort word pairs (edges) in decreasing order; \\
  \> {\bf Loop:} through edges; \\
  \> \>  Add top edge $e$ into the list of edges; \\
  \> \>  Adjust edge weights in sentences which contained edge $e$; \\ 
  \> {\bf Until} each sub-graph has a spanning tree (i.e. sentence has a parse);\\
  {\bf Output:} a set of spanning trees for all $S_i$;
\end{tabbing}
} \rule{1.0 \linewidth}{1pt} \label{alg_txt}
\end{table}

The preceding argument shows that the greedy algorithm has approximation ratio $O(\log n)$,
where $n$ is the length of the input. Table~1 presents description of the algorithm. 
Informally speaking, the algorithm could be explain as follows:
   {\em  Estimate potential number of sub-graphs each edge would contribute to if used.
   Loop through all edges, adding in (greedily) the edge which contributes to
the most spanning trees, then re-calculate potential contributions.}

\subsection{Contrast with Matroid Union}

Consider the matroid union problem for matroids $M_i^*$.
The matroid union problem is:
$$ \max \left\{\; \card{ \Union_i S_i} \,:\, S_i \in \cI_i^*  \;\right\} $$
But note that $S_i \in \cI_i^*$ iff $r_i(V \setminus S_i) = r_i(V)$.
In other words, $S_i \in \cI_i^*$ iff $\bar{S_i}$ contains a basis for $M_i$.
And maximizing the size of the union is the same as minimizing the size of the complement
of the union.
So an equivalent problem is:
$$ \min \left\{\; \card{ \Intersect_i \overline{S_i}} \,:\,  \mbox{$\overline{S_i}$ contains a basis for $M_i$} \;\right\} $$
The minimum does not change if we assume that $\overline{S_i}$ in fact \emph{is} a basis. 
So, letting $T_i$ denote $\bar{S_i}$, we obtain the equivalent problem:
$$ \min \left\{\; \card{ \Intersect_i T_i} \,:\, \text{$T_i$ is a basis for $M_i$} \;\right\} $$
This problem is solvable in polynomial time, because it is just matroid union in disguise.
It is quite similar to the Minimum Basis Hitting Set problem, except that it 
has an ``intersection'' rather than an ``union''.

\subsection{Empirical study}
We ran preliminary experiments with the approximation algorithm on adult child-directed speech from the CHILDES corpus~\cite{CHILDES}. These experiments demonstrated that the algorithm performs better than the baseline adjacency heuristic because of its ability to pick out non-adjacent dependencies. For example, the sentence "Is that a woof?" is parsed into the following set of links: woof-is, that-is, a-woof. The links correspond to the correct parse tree of the sentence, In contrast, the baseline adjacency heuristic would parse the sentence into is-that; that-a; and a-woof,  which fails to capture the dependence between the predicate noun "woof" and the verb, and postulates a non-existent dependency between the determiner "a" and the subject "that". However, more work is needed to thoroughly assess the performance. In particular, one problem for direct application is the presence of repeated words in the sentence. The current implementation avoids the issue of repeated words in its entirety, by filtering the input text. An alternative approach is to  erase the edges among repeated words from the original fully connected graph. This assumes that no word can be a dependent of itself, which might be a problem in some contexts (e.g. "I know that you know"). Related work which was not completed at the time of writing this manuscript seeks to incorporate adjacency as a soft linguistic constraint on the graph by increasing initial weight edges of adjacent words.

\section{Discussion}

We presented some theoretical results for a problem on graphs which is inspired by the unsupervised link grammar induction problem from linguistics. Numerous possible directions for the future work would include searching for more efficient approximation algorithms under various additional constraints on admissible spanning trees, as well as characterizing instances of the problem which could be solved efficiently. Another possible direction is allowing "ungrammatical" corpus as input, e.g. searching efficiently for partial solutions, where several sentences remain unparsed or not fully parsed. 
Another direction is to look for a solution to a directed graph analog of the problem considered here, which would require finding minimal set of arborescences and relate to the directed dependency parsing. 
One other question which remains open is an edge weighing scheme which would reflect syntactic consideration and particular language-related constraints, as in the so-called Optimality Theory~\cite{SavovaPhD}. 

Exploring relation of this problem to other application would be interesting. One such example could be an autonomous network design, where an objective is to efficiently design a network that must connect joint units of organizations which do not necessarily trust each other and want to maintain their own skeletal sub-network in case their partner's links fail.

\bibliography{Notes}
\bibliographystyle{abbrv}  % allrefs
\end{document}